\documentclass{elsarticle}

\usepackage{amsbsy,amsmath,amssymb}

\usepackage{graphics}
\usepackage{tikz}
\usepackage[normalem]{ulem}

\newcommand{\dg}{\ensuremath{m}}

\usepackage[pdftex,%
    colorlinks,%
    pdfauthor={Esparza, Ganty, Kiefer, Luttenberger},%
    hypertexnames=false,%
    pdftitle={Parikh's Theorem: A simple and direct construction}]{hyperref}

\def\set#1{{\left\{ #1 \right\}}}
\def\nats{{\mathbb{N}}}

\makeatletter
\begingroup \catcode `|=0 \catcode `[= 1
\catcode`]=2 \catcode `\{=12 \catcode `\}=12
\catcode`\\=12 |gdef|@xcomment#1\end{comment}[|end[comment]]
|endgroup
\def\@comment{\let\do\@makeother \dospecials\catcode`\^^M=10\def\par{}}
\def\begincomment{\@comment\@xcomment}
\makeatother

\newif
  \iflong
  \longtrue
\newif
  \ifshort
    \shortfalse

\newtheorem{theorem}{Theorem}[section]
\newtheorem{lemma}{Lemma}[section]
\newtheorem{proposition}{Proposition}[section]
\newtheorem{fact}{Fact}[section]
\newdefinition{definition}{Definition}[section]
\newproof{proof}{Proof}

\def\vec{\boldsymbol}
\def\by#1{\mathop{{\hbox{\setbox0=\hbox{$\scriptstyle{#1\quad}$}{$\buildrel{\>\>\;#1\quad}\over{\hbox to \wd0{\rightarrowfill}}$}}}}}

\newcommand{\step}[2]{#1\!\!~\Rightarrow~\!\!#2}

\begin{document}

\title{Parikh's Theorem: A simple and direct automaton construction}

\author[mu]{Javier Esparza}
\ead{esparza@model.in.tum.de}
\author[pg]{Pierre Ganty}
\ead{pierre.ganty@imdea.org}
\author[ox]{Stefan Kiefer}
\ead{stefan.kiefer@comlab.ox.ac.uk}
\author[mu]{Michael Luttenberger}
\ead{luttenbe@in.tum.de}

\address[mu]{Institut f\"ur Informatik, Technische Universit\"at M\"unchen, 85748 Garching, Germany}
\address[pg]{The IMDEA Software Institute, Madrid, Spain}
\address[ox]{Oxford University Computing Laboratory, Oxford, UK}

\begin{abstract}
Parikh's theorem states that the Parikh image of a context-free language is semilinear or,
equivalently, that every context-free language has the same Parikh image as some regular language.
We present a very simple construction that, given a context-free grammar, produces
a finite automaton recognizing such a regular language.
\end{abstract}

\maketitle

The {\em Parikh image} of a word $w$ over an alphabet
$\{a_1, \ldots, a_n\}$ is the vector $(v_1, \ldots, v_n) \in \nats^n$
such that $v_i$ is the number of occurrences of $a_i$ in $w$.
For example, the Parikh image of $a_1a_1a_2a_2$ over the alphabet $\{a_1, a_2, a_3\}$ is $(2,2,0)$.
The Parikh image of a language is the set of Parikh images of its words.
Parikh images are named after Rohit Parikh, who in 1966 proved a
classical theorem of formal language theory which also carries his name.
Parikh's theorem \cite{Parikh66} states that the Parikh image of any context-free language
is \emph{semilinear}. Since semilinear sets coincide with the Parikh images of regular languages,
the theorem is equivalent to the statement that every context-free language has the same
Parikh image as some regular language.
For instance, the language $\{a^nb^n \mid n \geq 0\}$
has the same Parikh image as $(ab)^*$. This statement is also often referred to as
Parikh's theorem, see e.g.~\cite{HK99}, and in fact it has been considered a more natural
formulation \cite{Pil73}.

Parikh's proof of the theorem, as many other subsequent proofs
\cite{Gre72,Pil73,Lee74,Gol77,HK99,AEI02}, is constructive: given a
context-free grammar $G$, the proof produces (at least implicitly) an automaton
or regular expression whose language has the same Parikh image as $L(G)$.
However, these constructions are relatively complicated, not given explicitly,
or yield crude upper bounds: automata of size $\mathcal{O}(n^n)$ for grammars in Chomsky
normal form with $n$ variables (see Section \ref{sec:related} for a detailed
discussion).  In this note we present an explicit and very simple construction
yielding an automaton with $\mathcal{O}(4^n)$ states, for a lower bound of $2^n$.
An application of the automaton is briefly discussed in Section \ref{sec:appl}:
the automaton can be used to algorithmically derive the semilinear set, and,
using recent results on Parikh images of NFAs \cite{ToThesis,KT10}, it leads to
the best known upper bounds on the size of the semilinear set for a given
context-free grammar.

\section{The Construction}

We follow the notation of \cite[Chapter 5]{HMU06}.
Let $G=(V,T,P,S)$ be a context-free grammar with a set~$V = \set{A_1, \ldots, A_n}$
of {\em variables} or {\em nonterminals}, a set $T$ of {\em terminals}, a set $P \subseteq V \times (V \cup T)^*$ of {\em productions}, and an {\em axiom} $S \in V$.
We construct a nondeterministic finite automaton (NFA) whose language has the same
Parikh image as~$L(G)$. The transitions of this automaton will be labeled with words of $T^*$,
but note that by adding intermediate states (when the words have length greater than one)
and removing $\epsilon$-transitions (i.e., when the words have length zero),
such an NFA can be easily brought in the more common form where transition labels are
elements of $T$.

We need to introduce a few notions.
For $\alpha \in (V \cup T)^*$ we denote by
$\Pi_V(\alpha)$ (resp.\ $\Pi_T(\alpha)$) the Parikh image of $\alpha$ where the
components not in $V$ (resp.\ $T$) have been projected away.  Moreover, let
$\alpha_{/V}$ (resp.\ $\alpha_{/T}$) denote the projection of $\alpha$ onto $V$ (resp.\ $T$).
For instance, if $V = \{A_1,A_2\}$, $T=\{a, b, c\}$,
and $\alpha = aA_2bA_1A_1$, then  $\Pi_V(\alpha) = (2,1)$, $\Pi_T(\alpha) = (1,1,0)$ and $\alpha_{/T}=ab$.
A pair $(\alpha, \beta) \in  (V \cup T)^*\times (V \cup T)^*$ is a {\em step}, denoted
by $\step{\alpha}{\beta}$, if there
exist $\alpha_1, \alpha_2 \in (V \cup T)^*$ and a production $A \rightarrow \gamma$ such that
$\alpha = \alpha_1A\alpha_2$ and $\beta=\alpha_1 \gamma \alpha_2$. Notice that given a step
$\step{\alpha}{\beta}$, the strings $\alpha_1, \alpha_2$ and the production $A \rightarrow \gamma$ are unique.
The {\em transition} associated to a step $\step{\alpha}{\beta}$ is the triple
$t(\step{\alpha}{\beta}) = (\Pi_V(\alpha),\gamma_{/T}, \Pi_V(\beta))$. For example, if $V=\{A_1,A_2,A_3\}$ and $T = \{a,b\}$,
then $t(\step{A_2aA_1}{A_2aA_2bA_3}) = ((1,1,0), b, (0,2,1))$.

\begin{definition}
Let $G=(V,T,P,S)$ be a context-free grammar and let $n=|V|$.
The {\em $k$-Parikh automaton} of $G$ is the NFA $M_G^k = (Q,T^*,\delta,q_0,\{q_f\})$ defined as follows:
\begin{itemize}
\item $Q = \set{(x_1,\dots,x_n)\in \nats^n \mid \sum_{i=1}^n x_i\leq k}$;
\item $\delta =  \{ t(\step{\alpha}{\beta}) \mid \mbox{ $\alpha \Rightarrow \beta$ is a step and $\Pi_V(\alpha), \Pi_V(\beta) \in Q$}  \}$;
\item $q_0 = \Pi_V(S)$;
\item $q_f = \Pi_V(\varepsilon) = (0, \ldots, 0)$.
\end{itemize}
\end{definition}
It is easily seen that $M_G^k$ has exactly \({n+k \choose n}\) states.

\noindent
Figure \ref{fig:example} shows the $3$-Parikh automaton of the context-free grammar
with productions $A_1~\rightarrow~A_1A_2 \vert a,  A_2~\rightarrow~bA_2aA_2 \vert cA_1$ and
axiom $A_1$. The states are all pairs $(x_1, x_2)$ such that $x_1 + x_2 \leq 3$. Transition
$(0,2) \by{ba} (0,3)$ comes e.g.\ from the step $A_2A_2 \Rightarrow
bA_2aA_2A_2$, and can be interpreted as follows: applying the production $A_2
\rightarrow bA_2aA_2$ to a word with zero occurrences of $A_1$ and two
occurrences of $A_2$ leads to a word with one new occurrence of $a$ and $b$,
zero occurrences of $A_1$, and three occurrences of $A_2$.

\begin{figure}[h]
    \centering
    \scalebox{0.6}{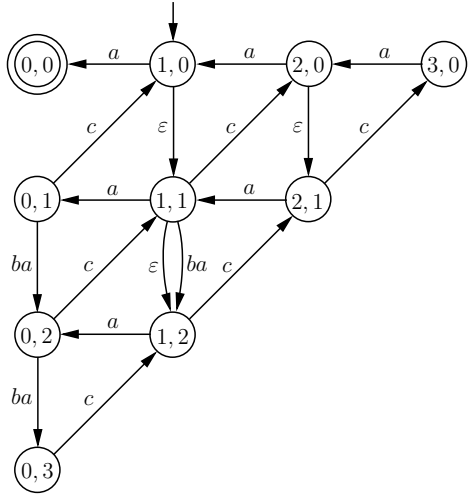}
    \caption{The $3$-Parikh automaton of
    $A_1 \rightarrow A_1A_2 \vert a,\ A_2 \rightarrow bA_2aA_2 \vert cA_1$ with $S = A_1$.
    }
    \label{fig:example}
\end{figure}

We define the {\em degree} of~$G$ by $\dg:=-1+\max\{ |\gamma_{/V}|\,:\, (A\to\gamma) \in P\}$;
 i.e., $\dg + 1$ is the maximal number of variables on the right hand sides. For instance, the degree
of the grammar in Fig.~\ref{fig:example} is $1$.
Notice that if $G$ is in Chomsky normal form then $\dg \le 1$, and
$\dg \le 0$ if{}f $G$ is regular.

In the rest of the note we prove:
\begin{theorem}
\label{thm:main}
 If $G$ is a context-free grammar with $n$ variables and degree~$\dg$,
  then $L(G)$ and $L(M_G^{n\dg+1})$ have the same Parikh image.
\end{theorem}

For the grammar of Figure \ref{fig:example} we have $n=2$ and $\dg=1$,
and Theorem \ref{thm:main} yields $L(G)=L(M_G^3)$. So
the language of the automaton of the figure has the same
Parikh image as the language of the grammar.

It is easily seen that $M_G^k$ has exactly ${n+k \choose k}$ states.
Using standard properties of binomial coefficients, for $M_G^{n\dg+1}$ and
$\dg \geq 1$ we get an upper bound of $2 \cdot (\dg+1)^n \cdot e^n$ states.
For $\dg \leq 1$ (e.g.\, for grammars in Chomsky normal form),
the automaton~$M_G^{n+1}$ has ${2n+1 \choose n} \leq 2^{2n+1} \in \mathcal{O}(4^n)$ states.
On the other hand, for every
$n\geq 1$ the grammar $G_n$ in Chomsky normal with
productions $\set{A_k \rightarrow A_{k-1} \ A_{k-1} \mid 2\leq k\leq n}\cup\set{A_1 \rightarrow a}$
and axiom $S=A_n$ satisfies $L(G_n)=\set{a^{2^{n-1}}}$, and therefore
the smallest Parikh-equivalent NFA has $2^{n-1}+1$ states.  This shows that our
construction is close to optimal.

\section{The Proof}
Given $L_1,L_2\subseteq T^*$, we write $L_1 =_{\Pi} L_2$ (resp.\ $L_1 \subseteq_{\Pi} L_2$),
 to denote that the Parikh image of~$L_1$ is equal to (resp.\ included in) the Parikh image of~$L_2$.
Also, given $w,w'\in T^*$, we abbreviate $\set{w}=_{\Pi}\set{w'}$ to $w=_{\Pi}w'$.

We fix a context-free grammar $G=(V,T,P,S)$ with $n$ variables and degree~$\dg$.
In terms of the notation we have just introduced, we have to prove $L(G) =_\Pi L(M_G^{n \dg + 1})$.
One inclusion is easy:

\begin{proposition}
For every $k \geq 1$ we have $L(M_G^k) \subseteq_\Pi L(G)$.
\label{prop:firstinclusion}
\end{proposition}
\begin{proof}
Let $k\geq 1$ arbitrary, and let $q_0 \by{\sigma} q$ be a run of $M_G^k$
on the word $\sigma \in T^*$. We first claim that there exists a step sequence
$S \Rightarrow^*\alpha$ satisfying $\Pi_V(\alpha) = q$ and
$\Pi_T(\alpha)=\Pi_T(\sigma)$.  The proof is by induction on the length $\ell$
of $q_0 \by{\sigma} q$. If $\ell=0$, then $\sigma = \varepsilon$, and we
choose $\alpha = S$, which satisfies $\Pi_V(S) = q_0$ and $\Pi_T(S) =
(0,\ldots,0) = \Pi_T(\varepsilon)$. If $\ell > 0$, then let $\sigma = \sigma'
\gamma$ and $q_0 \by{\sigma'} q' \by{\gamma} q$.  By induction
hypothesis 
there is a step sequence $S
\Rightarrow^*\alpha'$ satisfying $\Pi_V(\alpha') = q'$ and
$\Pi_T(\alpha')=\Pi_T(\sigma')$. Moreover, since $q' \by{\gamma} q$
is a transition of $M_G^k$, there is a production $A \rightarrow \gamma'$ and a
step $\step{\alpha_1 A \alpha_2}{\alpha_1 \gamma \alpha_2}$ such that
$\Pi_V(\alpha_1A\alpha_2) = q'$, $\Pi_V(\alpha_1\gamma'\alpha_2) =
q$ and $\gamma'_{/T}=\gamma$. Since $\Pi_V(\alpha') = q' =
\Pi_V(\alpha_1A\alpha_2)$, $\alpha'$ contains at least one occurrence of $A$,
i.e, $\alpha'=\alpha_1'A\alpha_2'$ for some $\alpha_1',\alpha_2'$. We choose
$\alpha = \alpha_1'\gamma'\alpha_2'$, and get $\Pi_V(\alpha)=\Pi_V(\alpha_1'\gamma'\alpha_2')
= \Pi_V(\alpha_1'A\alpha_2') - \Pi_V(A) + \Pi_V(\gamma') =\Pi_V(\alpha')-
\Pi_V(A) + \Pi_V(\gamma') = \Pi_V(\alpha_1 A \alpha_2)- \Pi_V(A) +
\Pi_V(\gamma')=\Pi_V(\alpha_1\gamma'\alpha_2) = q$.  Also
$\Pi_{T}(\alpha)=\Pi_T(\alpha_1'\gamma'\alpha_2')=\Pi_T(\alpha_1'A\alpha_2')+\Pi_T(\gamma')=
\Pi_T(\alpha') + \Pi_T(\gamma') = \Pi_T(\sigma') + \Pi_T(\gamma')
=\Pi_T(\sigma') + \Pi_T(\gamma)=\Pi_T(\sigma)$. This concludes the proof of the
claim.

Now, let $\sigma$ be an arbitrary word with $\sigma \in L(M_G^k)$.
Then there is a run $q_0 \by{\sigma} \Pi_V(\varepsilon)$.
By the claim there exists a step sequence $S \Rightarrow^*\alpha$ satisfying
 $\Pi_V(\alpha) = (0,\ldots,0)$ and $\Pi_T(\alpha)=\Pi_T(\sigma)$.
So $\alpha \in T^*$, and hence $\alpha \in L(G)$.
Since $\Pi_T(\alpha)=\Pi_T(\sigma)$ we have $\alpha =_\Pi \sigma$, and we are done.\qed
\end{proof}

The proof of the second inclusion $L(G) \subseteq_\Pi L(M_G^{n \dg +1})$ is more involved.
To explain its structure we need a definition.

\begin{definition}
A derivation $S=\alpha_0 \Rightarrow \cdots \Rightarrow \alpha_\ell$ of $G$
has index $k$ if for every $i\in\set{0,\ldots,\ell}$, the word $(\alpha_i)_{/V}$ has length
at most $k$. The set of words derivable through derivations of index
$k$ is denoted by $L_k(G)$.
\end{definition}
\noindent For example, the derivation $A_1 \Rightarrow A_1A_2 \Rightarrow A_1cA_1
\Rightarrow A_1ca \Rightarrow aca$ has index two.
Clearly, we have $L_1(G) \subseteq L_2(G) \subseteq L_3(G) \ldots$ and $L(G) = \bigcup_{k\geq 1} L_k(G)$.

The proof of  $L(G) \subseteq_\Pi L(M_G^{n\dg+1})$ is divided into two parts.
We first prove the {\em Collapse Lemma}, Lemma \ref{lem:collapse}, stating
that $L(G) \subseteq_\Pi L_{n\dg+1}(G)$, and then we prove, in Lemma~\ref{lem:deriv2parikhnfa},
that $L_k(G) \subseteq_\Pi L(M_G^k)$ holds for every $k \geq 1$.
A similar result has been proved in~\cite{EKL09:newtProgAn} with different notation and in a different context.
We reformulate its proof here for the reader interested in a self-contained proof.

\paragraph{The Collapse Lemma}  We need a few preliminaries.
We assume the reader is familiar with the
fact that every derivation can be parsed into a
{\em parse tree} \cite[Chapter 5]{HMU06}, whose  {\em yield}
is the word produced by the derivation. We denote the yield of a parse
tree $t$ by $Y(t)$, and the set of yields of a set $\mathcal{T}$ of trees by $Y(\mathcal{T})$.
Figure~\ref{fig:dtree} shows the parse
tree of the derivation $A_1 \Rightarrow A_1A_2 \Rightarrow aA_2 \Rightarrow abA_1 \Rightarrow aba$.
We introduce the notion of dimension of a parse tree.
\begin{figure}[h]
    \centering
    \scalebox{0.4}{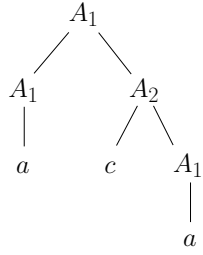}
        \caption{A parse tree of $A_1 \rightarrow A_1A_2 \vert a,\ A_2 \rightarrow bA_2aA_2 \vert cA_1$ with $S = A_1$}

\label{fig:dtree}

\end{figure}

\begin{definition} Let $t$ be a parse tree. A {\em child} of $t$ is a subtree of $t$ whose root
is a child of the root of $t$. A child of~$t$ is called \emph{proper} if its root is not a leaf, i.e.,
if it is labeled with a variable. The {\em dimension} $d(t)$ of a parse tree $t$ is inductively defined 
as follows. If $t$ has no proper children, then $d(t)=0$.
Otherwise, let $t_1,t_2,\ldots,t_r$ be the proper children of $t$ sorted such that $d(t_1)\ge d(t_2) \ge \ldots \ge d(t_r)$.
Then
 \[d(t)=\begin{cases}
          d(t_1) &\text{if $r=1$ or $d(t_1) > d(t_2)$} \\ d(t_1) + 1 &\text{if $d(t_1) = d(t_2)$.}
        \end{cases}
 \]
The set of parse trees of $G$ of dimension $k$ is denoted by
$\mathcal{T}^{(k)}$, and the set of all parse trees of $G$ by $\mathcal{T}$.
%
\end{definition}
\noindent The parse tree of Fig.~\ref{fig:dtree} has  two
children, both of them proper. It has dimension 1 and height 3.
Observe also the following fact, which can be easily proved by induction.
\begin{fact}
 Denote by~$h(t)$ the height of a tree~$t$.
 Then $h(t) > d(t)$.
\label{fact:height}
\end{fact}

For the proof of the collapse lemma, $L(G) \subseteq_\Pi L_{n\dg+1}(G)$,
 observe first that,
since every word in $L(G)$ is the yield of some parse tree, we have
$L(G) = Y(\mathcal{T})$, and so it suffices to show $Y(\mathcal{T}) \subseteq_\Pi L_{n\dg+1}(G)$.
The proof is divided into two parts.
We first show $Y(\mathcal{T}) \subseteq_\Pi \bigcup_{i=0}^n Y(\mathcal{T}^{(i)})$ in Lemma \ref{lem:dimension},
and then we show $\bigcup_{i=0}^n Y(\mathcal{T}^{(i)}) \subseteq L_{n\dg+1}(G)$ in
Lemma \ref{lem:derivation}. Actually, the latter proves the stronger result
that parse trees of dimension $k\geq 0$ have derivations of index $k\dg+1$, i.e.,
$Y(\mathcal{T}^{(k)}) \subseteq  L_{k\dg+1}(G)$ for all $k \leq 0$.

\begin{lemma}
    $Y(\mathcal{T}) \subseteq_\Pi \bigcup_{i=0}^nY(\mathcal{T}^{(i)})$.
    \label{lem:dimension}
\end{lemma}
\begin{proof}
In this proof we write $t=t_1 \cdot t_2$ to denote that $t_1$ is a parse
tree except that exactly one leaf $\ell$ is labelled by a variable, say $A$,
instead of a terminal; the tree $t_2$ is a parse tree with root $A$; and
the tree $t$ is obtained from $t_1$ and $t_2$ by replacing the leaf~$\ell$
of $t_1$ by the tree $t_2$. Figure~\ref{fig:dtree_decomp} shows an example.

\begin{figure}[h]
    \centering
    \scalebox{0.4}{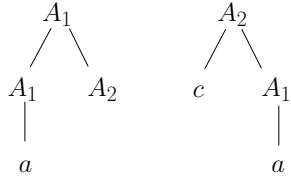}
    \caption{A decomposition $t_1$, $t_2$ such that $t=t_1\cdot t_2$ is the parse
        tree of Fig.~\ref{fig:dtree}}
    \label{fig:dtree_decomp}
\end{figure}

In the rest of the proof we abbreviate {\em parse tree} to {\em tree}.
We need to prove that for every tree $t$ there exists a
tree $t'$ such that $Y(t)=_{\Pi}Y(t')$ and $d(t') \leq n$. We shall
prove the stronger result that moreover $t$ and $t'$ have the 
same number of nodes, and the set of variables appearing in $t$ 
and $t'$ coincide.

Say that two trees $t, t'$ are {\em $\Omega$-equivalent} if 
they have the same number of nodes, the sets of variables appearing in $t$ 
and $t'$ coincide, and $Y(t)=_{\Pi}Y(t')$ holds. Say further that a tree 
$t$ is {\em compact} if $d(t) \le K(t)$, where
$K(t)$ denotes the number of variables that appear in $t$.
Since $K(t) \leq n$ for every $t$, it suffices to show that
every tree is $\Omega$-equivalent to a compact tree.
We describe a recursive ``compactification procedure'' ${\it Compact}(t)$ that 
transforms a tree $t$ into an $\Omega$-equivalent compact tree, and prove that
it is well-defined and correct. By well-defined we mean that some assumptions
made by the procedure about the existence of some objects indeed hold.

${\it Compact}(t)$ consists of the following steps: 
\begin{itemize}
\item[(1)] If $t$ is compact then return $t$ and terminate.
\item[(2)] If $t$ is not compact then
\begin{itemize}
\item[(2.1)] Let $t_1, \ldots, t_r$ be the proper children of $t$, $r \geq 1$.
\item[(2.2)] For every $1 \leq i \leq r$:\ \  $t_i:= {\it Compact}(t_i)$. \\
             (I.e., replace in $t$ the subtree $t_i$ by the result of compactifying $t_i$).\\
             Let $x$ be the smallest index $1 \leq x \leq r$ such that $K(t_x) = \max_i K(t_i)$.
\item[(2.3)] Choose an index $y \neq x$ such that $d(t_y) = \max_i d(t_i)$.
\item[(2.4)] Choose subtrees $t_x^a, t_x^b$ of $t_x$ and subtrees $t_y^a, t_y^b, t_y^c$ of $t_y$ such that 
\begin{itemize}
\item[(i)] $t_x = t^a_x \cdot t^b_x$ and $t_y = t^a_y \cdot( t^b_y \cdot t^c_y )$; and
\item[(ii)] the roots of $t^b_x, t^b_y$ and $t^c_y$ are labelled by the same variable.
\end{itemize}
\item[(2.5)] $t_x : = t^a_x\cdot(t^b_y\cdot t^b_x) \;\; ; \;\; t_y:=t^a_y\cdot t^c_y$.\\
(Loosely speaking, remove $t^b_y$ from $t_y$ and insert it into $t_x$.)
\item[(2.6)] Goto (1). 
\end{itemize}  
\end{itemize} 

We first prove that the assumptions at lines (2.1), (2.3), and (2.4) about
the existence of certain subtrees hold.\\

\noindent (2.1) {\it If $t$ is not compact, then $t$ has at least one proper child.}\\[0.1cm]
Assume that $t$ has no proper child. Then, by the definitions of
dimension and $K(t)$,  we have $d(t) = 0 \leq K(t)$, and so
$t$ is compact.\\

\noindent (2.3) {\it Assume that $t$ is not compact, has at least one proper child, and all its proper children 
are compact. Let $x$ be the smallest index $1 \leq x \leq r$ such that $K(t_x) = \max_i K(t_i)$. There 
there exists an index $y \neq x$ such that $d(t_y) = \max_i d(t_i)$.}\\[0.1cm]
Let $1 \leq y \leq r$ (where for the moment possibly $x = y$) be an index 
such that $d(t_y) = \max_id(t_i)$. We have
\begin{equation} \label{eq:new-sequence-1}
\begin{aligned}
 d(t) 
 & \le d(t_y) + 1 && \text{(by definition of dimension and of~$y$)} \\
 & \le K(t_y) + 1 && \text{(as $t_y$ is compact)} \\
 & \le K(t_x) + 1 && \text{(by definition of~$x$)} \\
 & \le K(t) + 1   && \text{(as $t_x$ is a child of~$t$)} \\
 & \le d(t)       && \text{(as $t$ is not compact),}
\end{aligned}
\end{equation}
so all inequalities in~\eqref{eq:new-sequence-1} are in fact equalities.
In particular, we have $d(t) = d(t_y) + 1$ and so, by the definitions of dimension and 
of~$y$, there exists $y' \ne y$ such that $d(t_{y'}) = d(t_y)$. Hence $x \ne y$ or $x \ne y'$,
and w.l.o.g.\ we can choose $y$ such that $y \ne x$.\\

\noindent (2.4) {\it Assume that $t$ is not compact, all its proper children are compact, and it has 
two distinct proper children $t_x, t_y$ such that $K(t_x) = \max_i K(t_i)$ and $d(t_y) = \max_i d(t_i)$.
There exist subtrees $t_x^a, t_x^b$ of $t_x$ and subtrees $ t_y^a, t_y^b, t_y^c$ of
$t_y$ satisfying conditions (i) and (ii).} \\[0.1cm]
By the equalities in \eqref{eq:new-sequence-1} we have $K(t_y) = d(t_y)$.
By Fact~\ref{fact:height} we have $d(t_y) < h(t_y)$. So $K(t_y) < h(t_y)$, and therefore some
path of $t_y$ from the root to a leaf visits at least two nodes labelled
with the same variable, say $A$. So $t_y$ can be factored into $t^a_y \cdot( t^b_y \cdot t^c_y )$ 
such that the roots of $t^b_y$ and $t^c_y$ are labelled by~$A$. Since by the equalities in
\eqref{eq:new-sequence-1} we also have $K(t) = K(t_x)$, every variable that appears in $t$ appears also in $t_x$, and so $t_x$ contains a node labelled by $A$. So $t_x$ can be factored
into $t_x = t^a_x \cdot t^b_x$ with the root of $t^b_x$ labelled by~$A$.\\

This concludes the proof that the procedure is well-defined. It remains to show that
it terminates and returns an $\Omega$-equivalent compact tree. We start by proving the following
lemma:\\

\noindent {\it If ${\it Compact}(t)$ terminates and returns a tree $t'$, then  
$t$ and $t'$ are $\Omega$-equivalent.} \\[0.1cm]
We proceed by induction on the number of calls to ${\it Compact}$
during the execution of ${\it Compact}(t)$. If ${\it Compact}$ is called only once, then
only line (1) is executed, $t$ is compact, no step modifies $t$, and we are done. 
Assume now that ${\it Compact}$ is called more than once.
The only lines that modify $t$ are (2.2) and (2.5). Consider first line (2.2.). 
By induction hypothesis, each call to ${\it Compact}(t_i)$ during the execution of ${\it Compact}(t)$
returns a compact tree $t_i'$ that is $\Omega$-equivalent to $t_i$. Let $t_1$ and $t_2$ be the 
values of $t$ before and after the execution of $t_i := {\it Compact}(t_i)$. Then $t_2$
is the result of replacing $t_i$ by $t_i'$ in $t_1$. By the definition of 
$\Omega$-equivalence, and since $t_i'$ is $\Omega$-equivalent to $t_i$, we get that 
$t_2$ is $\Omega$-equivalent to $t_1$. Consider now line (2.5), and let $t_1$ and $t_2$
 be the values of $t$ before and after the execution of $t_x : = t^a_x\cdot(t^b_y\cdot t^b_x)$ 
followed by the execution of $t_y:=t^a_y\cdot t^c_y$.
Since the subtree $t^b_y$ that is added to $t_x$ is subsequently removed 
from $t_y$, the Parikh-image of $Y(t)$, the number of nodes of $t$, 
and the set of variables appearing in $t$ do not change. This completes the proof of the lemma.\\

The lemma shows in particular that if the procedure terminates, then it returns an 
$\Omega$-equivalent tree. So it only remains to prove that the procedure always terminates.
Assume there is a tree $t$ such that ${\it Compact}(t)$ does not terminate. W.l.o.g.\ we 
further assume that $t$ has a minimal number of nodes. In this case all the calls to
line (2.2) terminate, and so the execution contains infinitely many steps that 
do not belong to any deeper call in the call tree, and in particular infinitely many executions
of the block (2.3)-(2.5). We claim that in all executions of this block the index 
$x$ {\em has the same value}. For this, observe first that, by the lemma, 
the execution of line (2.2) does not change the number of nodes or the set of variables occurring in 
each of $t_1, \ldots, t_r$.
In particular, it preserves the value of $K(t_1), \ldots, K(t_r)$.
Observe further that each time line (2.5) is executed, the procedure adds nodes
to $t_x$, and either does not change or removes nodes from any 
other proper children of $t$. 
In particular, the value of $K(t_x)$ does not decrease, and for every $i \neq x$ the value
of $K(t_i)$ does not increase. So at the next execution of the block the index $x$ of the former
execution is still the smallest index satisfying $K(t_x) = \max_i K(t_i)$. Now, since
$x$ has the same value at every execution of the block, each execution strictly decreases 
the number of nodes of some proper child $t_y$ different from $t_x$, and only increases the
number of nodes of $t_x$. This contradicts the 
fact that all proper children of $t$ have a finite number of nodes.
\qed
\end{proof}

\begin{lemma}
    For every $k\geq 0\colon Y(\mathcal{T}^{(k)}) \subseteq L_{k\dg+1}(G)$.
\label{lem:derivation}
\end{lemma}
\begin{proof}
\newcommand{\As}[1]{A^{(#1)}}%
In this proof we will use the following notation.
If $D$ is a derivation $\alpha_0 \Rightarrow \cdots \Rightarrow \alpha_\ell$ and $w, w' \in (V \cup T)^*$,
 then we define $w D w'$ to be the step sequence $w \alpha_0 w' \Rightarrow \cdots \Rightarrow w \alpha_\ell w'$.

Let $t$ be a parse tree such that $d(t)=k$.
We show that there is a derivation for~$Y(t)$ of index $k \dg + 1$.
We proceed by induction on the number of non-leaf nodes in~$t$.
In the base case, $t$ has no proper child.
Then we have $k = 0$ and $t$ represents a derivation $S \Rightarrow Y(t)$ of index~$1$.
For the induction step, assume that $t$ has $r \ge 1$ proper children $t_1,\ldots,t_r$
where the root of $t_i$ is assumed to be labeled by $\As{i}$;
i.e., we assume that the topmost level of~$t$ is induced by a rule
$S\rightarrow \gamma_0 \As{1} \gamma_1 \cdots \gamma_{r-1} \As{r}\gamma_r$ for $\gamma_i \in T^\ast$.
Note that $r-1\le \dg$.
By definition of dimension, at most one child $t_i$ has dimension~$k$, while the other children have dimension at most $k-1$.
W.l.o.g.\ assume $d(t_1) \le k$ and $d(t_2), \ldots,d(t_r) \le k-1$.
By induction hypothesis, for all $1 \le i \le r$ there is a derivation~$D_i$ for~$Y(t_i)$
 such that $D_1$ has index $k \dg + 1$, and $D_2, \ldots, D_r$ have index $(k-1) \dg + 1$.
Define, for each $1 \le i \le r$, the step sequence
 \[
  D_i' := \gamma_0 \As{1} \gamma_1 \cdots \gamma_{i-2} \As{i-1} \gamma_{i-1} D_i \gamma_i Y(t_{i+1}) \gamma_{i+1} \cdots \gamma_{r-1} Y(t_r) \gamma_r \,.
 \]
If the notion of index is extended to step sequences in the obvious way, then $D_1'$ has index $k \dg + 1$,
 and for $2 \le i \le r$, the step sequence $D_i'$ has index $(i-1) + (k-1) \dg + 1 \le k \dg + 1$.
By concatenating the step sequences $S \Rightarrow \gamma_0 \As{1} \gamma_1 \cdots \gamma_{r-1} \As{r}\gamma_r$ and
 $D_r, D_{r-1}, \ldots, D_1$ in that order,
 we obtain a derivation for $Y(t)$ of index $k \dg + 1$.
\qed
\end{proof}

Putting Lemma \ref{lem:derivation} and Lemma \ref{lem:dimension} together we obtain:

\begin{lemma}{\bf [Collapse Lemma]} $L(G) \subseteq_\Pi L_{n\dg+1}(G)$.
\label{lem:collapse}
\end{lemma}
\begin{proof}
\[\begin{array}{rcll}
L(G) & = & Y(\mathcal{T}) \\
& \subseteq_\Pi &  \bigcup_{i=0}^nY(\mathcal{T}^{(i)}) & \mbox{(Lemma \ref{lem:dimension})} \\
         & \subseteq &L_{n\dg+1}(G) & \mbox{(Lemma \ref{lem:derivation})}
\end{array}\]
\hfill\qed
\end{proof}

\begin{lemma}
For every $k\geq 1$: $L_k(G) \subseteq_\Pi L(M_G^{k})$.
\label{lem:deriv2parikhnfa}
\end{lemma}
\begin{proof}
    We show that if $S\Rightarrow^* \alpha$ is a prefix of a derivation of index $k$ then $M_G^k$ has a run $q_0 \xrightarrow{w} \Pi_V(\alpha)$ such that $w\in T^*$ and $\alpha_{/T}=_{\Pi} w$.
The proof is by induction on the length $i$ of the prefix.

{\bf $i=0$.} In this case $\alpha=S$, and since $q_0=\Pi_V(S)$ and
$S_{/T}=\varepsilon$ we are done.

{\bf $i>0$.} Since $S\Rightarrow^i \alpha$ there exist $\beta_1 A\beta_2\in
    (V\cup T)^*$ and a production $A\rightarrow \gamma$ such that
    $S\Rightarrow^{i-1}\beta_1 A \beta_2\Rightarrow \alpha$ and $\beta_1 \gamma
    \beta_2=\alpha$. By induction hypothesis, there exists a run of $M_G^k$
        such that $q_0\xrightarrow{w_1} \Pi_V(\beta_1 A \beta_2)$
        and ${(\beta_1 A \beta_2)}_{/T}=_{\Pi} w_1$.
    Then the definition of $M_G^k$ and the fact that $S\Rightarrow^i\alpha$ is of index $k$
        show that there exists a transition $(\Pi_V(\beta_1 A \beta_2),\gamma_{/T},\Pi_V(\alpha))$,
        hence we find that $q_0\xrightarrow{w_1\cdot \gamma_{/T}}\Pi_V(\alpha)$.
        Next we conclude from $(\beta_1 A \beta_2)_{/T} =_{\Pi} w_1$
        and $\alpha=\beta_1 \gamma \beta_2$ that $\alpha_{/T} =_{\Pi} w_1 \cdot \gamma_{/T}$ and we are done.

    Finally, if $\alpha\in T^*$ so that $S\Rightarrow^* \alpha$ is a derivation, then
        $q_0 \xrightarrow{w} \Pi_V(\alpha)=(0,\ldots,0)$ where $(0,\ldots,0)$ is an accepting state and $\alpha=\alpha_{/T} =_{\Pi} w$.\qed
\end{proof}

We now have all we need to prove the other inclusion.

\begin{proposition}
$L(G) \subseteq_{\Pi} L(M_G^{n\dg+1})$.
\label{prop:harddirection}
\end{proposition}
\begin{proof}
\[\begin{array}{rcll}
    L(G) & \subseteq_\Pi & L_{n\dg+1}(G) & \mbox{(Collapse Lemma)} \\
    & \subseteq_\Pi & L(M_G^{n\dg+1}) & \mbox{(Lemma \ref{lem:deriv2parikhnfa})}
\end{array}\]\qed
\end{proof}

\section{An Application: Bounding the Size of Semilinear Sets}
\label{sec:appl}

Recall that a set $S \subseteq \nats^k$, $k \geq 1$, is {\em linear} if there is an {\em offset}
$\vec{b} \in \nats^k$ and {\em periods} $\vec{p}_1, \ldots, \vec{p}_j \in \nats^k$ such that
$S = \{ \vec{b}+ \sum_{i=1}^j \lambda_i \vec{p}_i \mid \lambda_1, \ldots,
\lambda_j \in \nats \}$. A set is {\em semilinear} if it is the union of a finite number of
linear sets. It is easily seen that the Parikh image of a regular language is semilinear.
Procedures for computing the
semilinear representation of the language starting from a regular expression
or an automaton are well-known (see e.g.\ \cite{Pil73}). Combined with Theorem \ref{thm:main}
they provide an algorithm for computing the Parikh image of a context-free language.

Recently, To has obtained an upper bound on the size of the semilinear representation
of the Parikh image of a regular language (see Theorem 7.3.1 of \cite{ToThesis}):

\begin{theorem}
\label{thm:to}
Let $A$ be an NFA with $s$ states over an alphabet of $\ell$ letters. Then $\Pi(L(A))$ is a union
of $\mathcal{O}(s^{\ell^2 +3\ell+3} \, \ell^{4\ell+6})$ linear sets with at most $\ell$ periods; the maximum
entry of any offset is $\mathcal{O}(s^{3\ell+3} \, \ell^{4\ell+6})$, and the maximum entry of any period
is at most $s$.
\end{theorem}

Plugging Theorem \ref{thm:main} into Theorem \ref{thm:to}, we get the (to our
knowledge) best existing upper bound on the size of the semilinear set
representation of the Parikh image of a context-free language.
Let $G=(V,T,P,S)$ be a context-free grammar of degree $\dg$ with $n=|V|$ and $t=|T|$.
Let $p$ be the total number of occurrences of terminals
in the productions of $G$, i.e., $p = \sum_{X \rightarrow \alpha \in P} |\alpha_{/T}|$.
The number of states of $M_G^{n\dg+1}$ is ${ n + n\dg + 1 \choose n}$. Recall that
the transitions of $M_G^{n\dg+1}$ are labelled with words of the form $\gamma_{/T}$,
where $\gamma$ is the right-hand-side of some production. Splitting transitions, adding intermediate states, and then removing $\epsilon$-transitions yields an NFA
with ${ n + n\dg + 1 \choose n} \cdot p$ states. So we finally obtain for the parameters $s$ and $\ell$ in Theorem \ref{thm:to} the values $s := {n + n\dg + 1 \choose n} \cdot p$, and $\ell := t$.
This result (in fact a slightly stronger one) has
been used in \cite{EG11} to provide a polynomial algorithm for a language-theoretic problem
relevant for the automatic verification of concurrent programs.

\section{Conclusions and Related Work}
\label{sec:related}

For the sake of comparison we will assume throughout this section that all
grammars have degree \(\dg \leq 1\).  Given $G$ a context-free grammar with $n$
variables, we have shown how to construct an NFA~$M$ with $\mathcal{O}(4^n)$
states such that $L(G)$ and $L(M)$ have the same Parikh image. We compare this
result with previous proofs of Parikh's theorem.

Parikh's proof~\cite{Parikh66} (essentially the same proof is given in~\cite{Sal73})
shows how to obtain a Parikh-equivalent regular expression from a
finite set of parse trees of $G$.
The complexity of the resulting construction is not studied.
By its definition, the regular expression basically consists of the sum of words
obtained from the parse trees of height at most $n^2$.
This leads to the admittedly rough bound that the regular expression consists of at most $\mathcal{O}(2^{2^{n^2-1}})$ words each of length at most $\mathcal{O}(2^{n^2})$.

Greibach~\cite{Gre72} shows that a particular substitution operator on
language classes preserves semilinearity of the languages. This result
implies Parikh's theorem, if the substitution operator is applied to the
class of regular languages. It is hard to extract a construction from this
proof, as it relies on previously proved closure properties of language classes.

Pilling's proof~\cite{Pil73} (also given in~\cite{Con71}) of Parikh's theorem uses algebraic properties of commutative regular languages.
From a constructive point of view, his proof leads to a procedure
that iteratively replaces a variable of the grammar $G$ by a regular
expression over the terminals and the other variables.
This procedure finally generates a regular expression which is Parikh-equivalent to~$L(G)$. Van Leeuwen~\cite{Lee74} extends Parikh's theorem to other
language classes, but, while using very different concepts and terminology,
his proof leads to the same construction as Pilling's.
Neither \cite{Pil73} nor \cite{Lee74} study the size of the resulting
regular expression.

Goldstine~\cite{Gol77} simplifies Parikh's original proof.
An explicit construction can be derived from the proof, but it is involved:
for instance, it requires to compute for each subset of variables, the computation
of all derivations with these variables up to a certain size depending
on a pumping constant.

Hopkins and Kozen~\cite{HK99} generalize Parikh's theorem to commutative Kleene algebra.
Like in Pilling~\cite{Pil73} their procedure to compute a Parikh-equivalent regular expression is iterative;
 but rather than eliminating one variable in each step, they treat all variables in a symmetric way.
Their construction can be adapted to compute a Parikh-equivalent finite automaton.
Hopkins and Kozen show (by algebraic means)
that their iterative procedure terminates after $\mathcal{O}(3^n)$ iterations for
a grammar with $n$ variables. In~\cite{EKL09:newtProgAn} we reduce this bound
(by combinatorial means) to~$n$ iterations.
The construction yields an automaton, but it is much harder to explain than ours.
The automaton has size~$\mathcal{O}(n^n)$.

In~\cite{AEI02} Parikh's theorem is derived from a small set of
 purely equational axioms involving fixed points.
It is hard to derive a construction from this proof.

In \cite{esp97} Parikh's theorem is proved by showing that the Parikh image of
a context-free language is the union of the sets of solutions of a finite
number of systems of linear equations. In \cite{seidl05} the theorem is also
implicitly proved, this time by showing that the Parikh image is the set of
models of an existential formula of Presburger arithmetic. While the
constructions yielding the systems of equations and the Presburger formulas are
very useful, they are also more complicated than our construction of the Parikh
automaton. Also, neither \cite{esp97} nor \cite{seidl05} give bounds on the
size of the semilinear set.

\section*{Acknowledgments}
We thank two anonymous referees for very useful suggestions.

\end{document}